\newif\ifSHORTVERSION
\SHORTVERSIONfalse

\newif\ifZENODOVERSION
\ZENODOVERSIONfalse

\ifSHORTVERSION
 \documentclass[letterpaper,10pt,conference]{ieeeconf}
 \IEEEoverridecommandlockouts
\else
 \documentclass[journal,10pt,twoside]{IEEEtran}
\fi

\usepackage[T1]{fontenc}
\usepackage[utf8]{inputenc}

\ifSHORTVERSION
 \usepackage{amsmath,amssymb,amsfonts,mathtools}

 \usepackage{amsthm}
\else
 \usepackage{amsmath,amssymb,amsfonts,amsthm,mathtools}
\fi

\usepackage{newtxtext,newtxmath}
\usepackage{microtype}
\usepackage{bm}
\let\labelindent\relax
\usepackage{enumitem}
\usepackage{aliascnt}

\ifSHORTVERSION
\else
 \usepackage[a4paper,left=18mm,right=18mm,top=24mm,bottom=22mm,columnsep=6mm]{geometry}
\fi

\ifZENODOVERSION
    \usepackage{eso-pic}
    \newcommand{\ZenodoSideNotice}{%
      \AddToShipoutPictureFG*{%
        \AtPageLowerLeft{%
          \put(
            \LenToUnit{5mm},
            \LenToUnit{35mm}
          ){%
            \rotatebox{90}{%
              \href{https://doi.org/10.5281/zenodo.23301985}{%
                \color{gray}%
                \sffamily\huge
                Zenodo preprint, 11 October 2026.
                DOI: 10.5281/zenodo.23301985%
              }%
            }%
          }%
        }%
      }%
    }
\else
\fi

\newtheorem{definition}{Definition}
\newaliascnt{proposition}{definition}
\newtheorem{proposition}[proposition]{Proposition}
\aliascntresetthe{proposition}
\newaliascnt{lemma}{definition}
\newtheorem{lemma}[lemma]{Lemma}
\aliascntresetthe{lemma}
\newaliascnt{theorem}{definition}
\newtheorem{theorem}[theorem]{Theorem}
\aliascntresetthe{theorem}
\newaliascnt{corollary}{definition}
\newtheorem{corollary}[corollary]{Corollary}
\aliascntresetthe{corollary}
\newaliascnt{assumption}{definition}
\newtheorem{assumption}[assumption]{Assumption}
\aliascntresetthe{assumption}
\theoremstyle{remark}
\newaliascnt{remark}{definition}
\newtheorem{remark}[remark]{Remark}
\aliascntresetthe{remark}
\newaliascnt{example}{definition}

\aliascntresetthe{example}

\usepackage{cite}
\usepackage[hidelinks]{hyperref}
\usepackage[nameinlink,noabbrev]{cleveref}
\crefformat{equation}{#2(#1)#3}
\Crefformat{equation}{Equation~#2(#1)#3}
\crefname{figure}{Fig.}{Figs.}
\Crefname{figure}{Figure}{Figures}
\crefname{table}{Table}{Tables}
\Crefname{table}{Table}{Tables}
\crefname{section}{Sec.}{Secs.}
\Crefname{section}{Section}{Sections}
\crefname{subsection}{Sec.}{Secs.}
\Crefname{subsection}{Section}{Sections}
\crefname{definition}{Definition}{Definitions}
\Crefname{definition}{Definition}{Definitions}
\crefname{proposition}{Proposition}{Propositions}
\Crefname{proposition}{Proposition}{Propositions}
\crefname{lemma}{Lemma}{Lemmas}
\Crefname{lemma}{Lemma}{Lemmas}
\crefname{theorem}{Theorem}{Theorems}
\Crefname{theorem}{Theorem}{Theorems}
\crefname{corollary}{Corollary}{Corollaries}
\Crefname{corollary}{Corollary}{Corollaries}
\crefname{assumption}{Assumption}{Assumptions}
\Crefname{assumption}{Assumption}{Assumptions}
\crefname{remark}{Remark}{Remarks}
\Crefname{remark}{Remark}{Remarks}
\crefname{example}{Example}{Examples}
\Crefname{example}{Example}{Examples}

\ifSHORTVERSION
\else
 \setlist[enumerate]{leftmargin=*,itemsep=0.3ex,topsep=0.5ex,parsep=0pt}
 \setlist[itemize]{leftmargin=*,itemsep=0.3ex,topsep=0.5ex,parsep=0pt}
 \usepackage{titlesec}
 \renewcommand{\thesection}{\arabic{section}}
 \renewcommand{\thesubsection}{\thesection.\arabic{subsection}}
 \renewcommand{\thesubsubsection}{\thesubsection.\arabic{subsubsection}}
 \titleformat{\section}{\normalfont\large\bfseries}{\thesection}{0.65em}{}
 \titleformat{\subsection}{\normalfont\normalsize\bfseries}{\thesubsection}{0.65em}{}
 \titleformat{\subsubsection}{\normalfont\normalsize\itshape}{\thesubsubsection}{0.65em}{}
 \titlespacing*{\section}{0pt}{2.0ex plus 0.5ex minus 0.2ex}{0.8ex}
 \titlespacing*{\subsection}{0pt}{1.5ex plus 0.4ex minus 0.2ex}{0.5ex}
 \titlespacing*{\subsubsection}{0pt}{1.2ex plus 0.3ex minus 0.2ex}{0.4ex}
\fi

\newif\ifanonymous
\anonymousfalse

\newcommand{\papertitle}{\huge Aerodynamic Promptness of Tilting-Propeller Multirotors}
\newcommand{\PaperAuthorsAffiliated}{Antonio Franchi\textsuperscript{1,2}}
\newcommand{\AffiliationOne}{\textsuperscript{1}Robotics and Mechatronics, Faculty of Electrical Engineering, Mathematics and Computer Science, University of Twente, Enschede, The Netherlands.}
\newcommand{\AffiliationTwo}{\textsuperscript{2}Department of Computer, Control and Management Engineering, Sapienza University of Rome, Rome, Italy.}
\newcommand{\CorrespondingAuthor}{\textit{Corresponding author:} Antonio Franchi, \href{mailto:schol@r-franchi.eu}{schol@r-franchi.eu}.}
\newcommand{\FundingStatement}{This work was partially funded by the Horizon Europe research project AUTOASSESS under grant agreement No.~101120732.}
\newcommand{\PaperAbstract}{This work develops a geometric theory of aerodynamic promptness for multirotor platforms whose propeller directions depend on a smooth orientation-mechanism configuration. The actuator state consists of the actual signed propeller speeds and the mechanism configuration, and the nonlinear actuator map retains both the signed-quadratic propeller law and the configuration-dependent thrust directions. Constant direction maps represent fixed propellers, while shared dependence on mechanism coordinates represents mechanically coupled tilts. A Riemannian metric on the product actuator manifold compares propeller accelerations and orientation-mechanism rates through physically interpretable scales. Composing the actuator-map differential, the musical isomorphism induced by this metric, and the dual differential defines a task-rate authority map from wrench covectors to wrench-rate vectors. In fixed task coordinates, the square root of the determinant of its matrix defines aerodynamic promptness. Positive promptness is equivalent to local wrench-rate surjectivity, while the associated spectrum resolves directional authority and conditioning. The same construction yields the metric horizontal inverse, the orthogonal projector onto wrench-fiber directions, and task-preserving internal motions that regulate promptness without changing the instantaneous wrench. The construction is related intrinsically to dynamic feedback linearization and specialized to the independently tilting, radially arranged Holocopter and Voliro architectures, clarifying how their established allocation and rank constructions fit within the same actuator-manifold geometry.}
\newcommand{\PaperKeywords}{Aerodynamic promptness, control allocation, differential geometry, dynamic feedback linearization, multirotor vehicles, tilting propellers, thrust-vectoring mechanisms}

\title{\papertitle}
\ifSHORTVERSION
 \ifanonymous
  \author{Anonymous Author}
 \else
  \author{\PaperAuthorsAffiliated\thanks{\AffiliationOne}\thanks{\AffiliationTwo}\thanks{\FundingStatement}}
 \fi
\else
 \ifanonymous
  \author{Anonymous Author}
 \else
  \author{\PaperAuthorsAffiliated}
 \fi
\fi

\ifanonymous
 \hypersetup{pdftitle={\papertitle},pdfauthor={Anonymous Author}}
\else
 \hypersetup{pdftitle={\papertitle},pdfauthor={Antonio Franchi}}
\fi

\ifSHORTVERSION
 \makeatletter
 \@ifundefined{keywords}{\newenvironment{keywords}{\par\smallskip\noindent\textbf{Keywords:}\ }{\par\smallskip}}{}
 \makeatother
\fi

\ifSHORTVERSION
\else
 \makeatletter
 \newcommand{\MakeNeutralTitle}{%
  \twocolumn[{%
   \begin{@twocolumnfalse}
    \vspace*{-1.5em}
    {\raggedright\fontsize{17}{20}\selectfont\bfseries\papertitle\par}
    \vspace{0.9em}
    \ifanonymous
     {\raggedright\normalsize\textbf{Anonymous Author}\par}
    \else
     {\raggedright\normalsize\textbf{\PaperAuthorsAffiliated}\par}
     \vspace{0.35em}
     {\raggedright\small\AffiliationOne\par\AffiliationTwo\par\CorrespondingAuthor\par\FundingStatement\par}
    \fi
    \vspace{1.1em}
    \noindent{\small\bfseries Abstract}\par
    \vspace{0.25em}
    {\small\noindent\PaperAbstract\par}
    \vspace{0.8em}
    {\small\noindent\textbf{Keywords:} \PaperKeywords.\par}
    \vspace{1.6em}
   \end{@twocolumnfalse}
  }]%
 }
 \makeatother
\fi

\newcommand{\RealNumbers}{\mathbb{R}}
\newcommand{\NonnegativeReals}{\mathbb{R}_{\geq 0}}
\newcommand{\PositiveReals}{\mathbb{R}_{>0}}
\newcommand{\IntegerNumbers}{\mathbb{Z}}
\newcommand{\UnitSphere}{\mathbb{S}^{2}}
\newcommand{\ActuatorStateSpace}{\mathcal{V}}
\newcommand{\AdmissibleActuatorStateSpace}{\mathcal{V}_{\mathrm{ad}}}
\newcommand{\PropellerSpeedSpace}{\mathcal{S}}
\newcommand{\OrientationConfigurationSpace}{\mathcal{Q}}
\newcommand{\configurationDimension}{n_{q}}
\newcommand{\WrenchSpace}{\mathcal{W}}
\newcommand{\InputSpace}{\mathcal{U}}
\newcommand{\RegularActuatorSet}{\mathcal{V}_{\mathrm{reg}}}
\newcommand{\actuatorCount}{n}
\newcommand{\wrenchDimension}{m}
\newcommand{\fullWrenchDimension}{6}
\newcommand{\spatialDimension}{3}
\newcommand{\actuatorIndex}{i}
\newcommand{\actuatorIndexSet}{\mathcal{I}}
\newcommand{\ActuatorTangentSpaceAt}[1]{T_{#1}\ActuatorStateSpace}
\newcommand{\WrenchTangentSpaceAt}[1]{T_{#1}\WrenchSpace}
\newcommand{\ActuatorCotangentSpaceAt}[1]{T_{#1}^{*}\ActuatorStateSpace}
\newcommand{\WrenchCotangentSpaceAt}[1]{T_{#1}^{*}\WrenchSpace}
\newcommand{\identityMatrix}[1]{I_{#1}}
\newcommand{\transpose}{\mathsf{T}}
\newcommand{\matrixRank}{\operatorname{rank}}
\newcommand{\matrixDeterminant}{\operatorname{det}}
\newcommand{\kernel}{\operatorname{ker}}
\newcommand{\image}{\operatorname{im}}
\newcommand{\diagOperator}{\operatorname{diag}}
\newcommand{\absoluteValue}[1]{\left\lvert#1\right\rvert}
\newcommand{\vectorNorm}[1]{\left\lVert#1\right\rVert}
\newcommand{\setBuilder}[2]{\left\{#1\,\middle|\,#2\right\}}
\newcommand{\differential}{\mathrm{d}}
\newcommand{\directSum}{\oplus}
\newcommand{\orthogonalComplement}{\perp}
\newcommand{\argmaximum}{\operatorname*{arg\,max}}
\newcommand{\BodyFrame}{\mathcal{F}_{B}}
\newcommand{\PropellerFrame}[1]{\mathcal{F}_{P_{#1}}}
\newcommand{\BodyOrigin}{O_{B}}
\newcommand{\PropellerOrigin}[1]{O_{P_{#1}}}
\newcommand{\propellerPosition}[1]{\bm{r}_{#1}}
\newcommand{\tiltAxis}[1]{\bm{a}_{#1}}
\newcommand{\nominalThrustDirection}[1]{\bm{d}_{#1}^{0}}
\newcommand{\propellerDirectionMap}[1]{\bm{d}_{#1}}
\newcommand{\rotationAboutAxis}[2]{\mathcal{R}(#1,#2)}
\newcommand{\crossProductMatrix}[1]{[#1]_{\times}}
\newcommand{\bodyWrench}{\bm{w}}
\newcommand{\armLength}{L}
\newcommand{\canonicalAxis}[1]{\bm{e}_{#1}}
\newcommand{\actuatorState}{\bm{v}}
\newcommand{\propellerSpeedArray}{\bm{s}}
\newcommand{\propellerSpeed}[1]{s_{#1}}
\newcommand{\propellerAccelerationArray}{\dot{\bm{s}}}
\newcommand{\configurationState}{\bm{q}}
\newcommand{\configurationRate}{\dot{\bm{q}}}
\newcommand{\tiltAngleArray}{\bm{\alpha}}
\newcommand{\tiltAngle}[1]{\alpha_{#1}}
\newcommand{\actuatorRate}{\dot{\bm{v}}}
\newcommand{\thrustMagnitude}[1]{T_{#1}}
\newcommand{\dragMomentMagnitude}[1]{M_{#1}}
\newcommand{\thrustCoefficient}[1]{k_{f,#1}}
\newcommand{\dragCoefficient}[1]{k_{m,#1}}
\newcommand{\spinDirection}[1]{\epsilon_{#1}}
\newcommand{\signedSquare}[1]{#1\absoluteValue{#1}}
\newcommand{\SignedSquareMap}{\psi}
\newcommand{\SignedSquareMapDerivative}{\psi'}
\newcommand{\unitWrenchCurve}[1]{\bm{b}_{#1}}
\newcommand{\actuatorMap}{f}
\newcommand{\actuatorMapDifferentialAt}[1]{\differential\actuatorMap_{#1}}
\newcommand{\actuatorMapJacobian}{J_{f}}
\newcommand{\speedJacobian}{J_{s}}
\newcommand{\configurationJacobian}{J_{q}}
\newcommand{\taskRate}{\dot{\bm{w}}}
\newcommand{\metric}{g}
\newcommand{\metricMatrix}{G}
\newcommand{\speedMetricMatrix}{G_{s}}
\newcommand{\configurationMetricMatrix}{G_{q}}
\newcommand{\couplingMetricMatrix}{G_{sq}}
\newcommand{\speedCapabilityMatrix}{A_{s}}
\newcommand{\configurationCapabilityMatrix}{A_{q}}
\newcommand{\wrenchFiber}{\mathcal{F}}
\newcommand{\fiberAt}[1]{\wrenchFiber_{#1}}
\newcommand{\verticalSpace}{\mathsf{V}}
\newcommand{\verticalSpaceAt}[1]{\verticalSpace_{#1}}
\newcommand{\horizontalSpace}{\mathsf{H}}
\newcommand{\horizontalSpaceAt}[1]{\horizontalSpace_{#1}}
\newcommand{\musicalSharp}{\sharp}
\newcommand{\authorityMap}{\mathscr{D}}
\newcommand{\horizontalInverseMap}{j}
\newcommand{\horizontalInverseMapAt}[1]{\horizontalInverseMap_{#1}^{\metric}}
\newcommand{\horizontalInverseMatrix}{J_{f}^{\#_{\metric}}}
\newcommand{\verticalProjectorMap}{p_{\mathrm{ver}}}
\newcommand{\verticalProjectorMapAt}[1]{{\verticalProjectorMap}_{#1}}
\newcommand{\verticalProjectorMatrix}{P_{\mathrm{ver}}}
\newcommand{\verticalActuatorRate}{\actuatorRate_{\mathrm{ver}}}
\newcommand{\taskMetric}{h}
\newcommand{\taskMetricMatrix}{H}
\newcommand{\taskRateAuthorityForm}{\mathcal{D}}
\newcommand{\taskRateAuthorityMatrix}{D}
\newcommand{\speedAuthorityContribution}{D_{s}}
\newcommand{\configurationAuthorityContribution}{D_{q}}
\newcommand{\aerodynamicPromptness}{\rho}
\newcommand{\metricWeightedPromptness}{\rho_{h}}
\newcommand{\promptnessSquared}{\sigma}
\newcommand{\promptnessReference}{\rho_{\mathrm{d}}}
\newcommand{\riemannianGradient}{\operatorname{grad}_{\metric}}
\newcommand{\desiredTaskRate}{\dot{\bm{w}}_{\mathrm{d}}}
\newcommand{\internalMotion}{\bm{\xi}}
\newcommand{\promptnessGain}{k_{\rho}}
\newcommand{\actuatorVelocityInput}{\bm{u}}
\newcommand{\actuatorDrift}{\bm{d}}
\newcommand{\actuatorInputMap}{b}
\newcommand{\actuatorInputMapAt}[1]{\actuatorInputMap_{#1}}
\newcommand{\actuatorInputMatrix}{B}
\newcommand{\inputMetric}{r}
\newcommand{\inputMetricMatrix}{R}
\newcommand{\inputMusicalSharp}{\sharp_{\inputMetric}}
\newcommand{\inputActuatorAuthorityMap}{c}
\newcommand{\inputActuatorAuthorityMapAt}[1]{\inputActuatorAuthorityMap_{#1}}
\newcommand{\inputActuatorAuthorityMatrix}{C}
\newcommand{\decouplingMap}{\Delta}
\newcommand{\decouplingMapAt}[1]{\decouplingMap_{#1}}
\newcommand{\decouplingMatrix}{A_{\mathrm{d}}}
\newcommand{\dflAuthorityMap}{d_{\mathrm{u}}}
\newcommand{\dflAuthorityMapAt}[1]{{\dflAuthorityMap}_{#1}}
\newcommand{\dflAuthorityMatrix}{D_{\mathrm{d}}}
\newcommand{\dflPromptness}{\rho_{\mathrm{dfl}}}
\newcommand{\driftCompensatedTaskRate}{\bm{y}}
\newcommand{\dflHorizontalInverse}{j_{\mathrm{u}}}
\newcommand{\dflHorizontalInverseAt}[1]{{\dflHorizontalInverse}_{#1}}
\newcommand{\TiltTorus}{\mathbb{T}^{\actuatorCount}}
\newcommand{\tiltTransmissionMap}{\alpha}
\newcommand{\holocopterRotorCount}{4}
\newcommand{\voliroDecompositionMap}{\chi}
\newcommand{\voliroRotorCount}{6}
\newcommand{\radialAzimuth}[1]{\beta_{#1}}
\newcommand{\voliroVerticalForce}[1]{F_{\mathrm{v},#1}}
\newcommand{\voliroLateralForce}[1]{F_{\mathrm{l},#1}}
\newcommand{\voliroDecomposedForceArray}{\bm{F}_{\mathrm{dec}}}
\newcommand{\voliroStaticAllocationMap}{a_{\mathrm{static}}}
\newcommand{\voliroStaticAllocationMatrix}{A_{\mathrm{static}}}

\begin{document}
\ifZENODOVERSION
    \ZenodoSideNotice
\else
\fi

\ifSHORTVERSION
 \maketitle
 \begin{abstract}\PaperAbstract\end{abstract}
\else
 \pagestyle{plain}
 \MakeNeutralTitle
\fi

\section{Introduction}
\label{sec:introduction}

Multirotor aerial robots combine mechanical simplicity, vertical takeoff and landing, and precise hovering, but conventional coplanar designs couple translational and rotational motion because their propellers generate force predominantly along one body-fixed direction~\cite{Mahony2012Multirotor}. This limitation becomes consequential in aerial manipulation, contact inspection, and operation in constrained environments, where the vehicle may need to control position, attitude, and interaction wrench simultaneously~\cite{Ruggiero2018AerialManipulation,Zhao2022Dragon}.

Fully actuated multirotors address this limitation by arranging fixed propellers in nonparallel directions or by actively vectoring their thrust. The resulting design space includes optimized fixed-tilt platforms~\cite{Park2016ODAR,Brescianini2018Omnidirectional}, variable-tilt quadrotors and hexarotors~\cite{Ryll2015Holocopter,Kamel2018Voliro}, articulated vectorable-thrust systems~\cite{Zhao2018Dragon}, and recent tiltable platforms integrating allocation, trajectory optimization, and physical interaction~\cite{Ruggia2025MOMAV,Li2026Effector}. The breadth of these architectures and their allocation properties is surveyed in~\cite{Rashad2020Review,Hamandi2021Taxonomy}.

Full actuation determines which instantaneous body wrenches can be generated. Redundant actuation introduces a different question: which actuator state should realize a wrench when many states produce the same output? Classical control allocation resolves this ambiguity through generalized inverses, constrained optimization, or preferred-control objectives~\cite{Durham2017Allocation}. In fully actuated aerial vehicles, these methods have been extended through null-space allocation, explicit actuator constraints, and actuator-aware prediction~\cite{Su2021Nullspace,Santos2022Allocation,Pretto2026Nullspace}. Their prevailing objectives concern effort, feasibility, smoothness, singularity avoidance, or actuator margins.

A complementary line of work quantifies how configuration shapes
differential task authority. Manipulability introduced volume and
conditioning measures for robot velocity transmission
in~\cite{Yoshikawa1985Manipulability}, while nonlinear control
identifies surjectivity of the relevant decoupling map as the
feasibility boundary for feedback linearization
in~\cite{Isidori1995}. Aerodynamic promptness was introduced
in~\cite{Franchi2026Promptness} as a measure of the wrench-rate
readiness of a redundant multirotor actuator state, together with a
geometric characterization of its fiberwise trade-off with energetic
expenditure. The drag-aware aerodynamic manipulability framework
in~\cite{Franchi2026AeroPromptness} subsequently incorporated motor
torque limits and aerodynamic drag into this readiness geometry.
The subsequent variable-aerodynamic-damping study identified the
complementary passive effect: motion along a constant-wrench fiber can
also shape incremental aerodynamic damping without changing the trim
wrench~\cite{Franchi2026VADA}. Building on the promptness, or
readiness, perspective, later work formulated a
control-barrier-function lower bound on motor-only wrench-rate
authority~\cite{Silano2026ReadinessBarrier} and used active rotor tilt
to preserve this certified authority under actuator
limitations~\cite{Silano2026TiltResource}.

This work develops the intrinsic differential geometry of aerodynamic promptness for multirotors with configuration-dependent propeller directions. The physical actuator state is modeled on the product of a propeller-speed domain and the configuration manifold of the propeller-orientation mechanism. Constant direction maps represent fixed propellers, while independent or shared configuration dependence represents independently or mechanically coupled tilting propellers. A Riemannian metric assigns a common normalized meaning to propeller accelerations and orientation-mechanism rates. The actuator-map
differential and the musical isomorphism induced by this metric then
define complete actuator-state wrench-rate authority independently of
actuator coordinates. This construction extends the fiberwise
promptness and readiness framework
of~\cite{Franchi2026Promptness,Franchi2026AeroPromptness} to
configuration-dependent propeller orientation and connects it to wrench-fiber
geometry, dynamic feedback linearization, and task-preserving internal
motion. The present work concerns active wrench-rate authority; the
passive damping mechanism of~\cite{Franchi2026VADA} is a complementary
fiber-dependent property rather than part of the promptness
definition.

The contributions are as follows.
\begin{enumerate}
\item We formulate the physical actuator-to-wrench map of a multirotor directly in actual signed propeller speeds and the configuration of a smooth propeller-orientation mechanism, thereby including fixed, independently tilting, and mechanically coupled propellers.
\item We define a smooth actuator-rate metric that compares propeller accelerations and orientation-mechanism rates through physically interpretable characteristic capabilities.
\item We construct the task-rate authority map intrinsically, define aerodynamic promptness from its coordinate volume, and separate the propeller-acceleration and orientation-mechanism contributions under a block-diagonal metric.
\item We prove that positive promptness is equivalent to local wrench-rate surjectivity and establish the corresponding minimum-norm horizontal inverse and orthogonal wrench-fiber projector.
\item We relate the construction to actuator dynamics and dynamic feedback linearization by distinguishing the intrinsic decoupling map, its coordinate matrix, and the input-induced authority map.
\item We derive task-preserving internal motions that regulate promptness along wrench fibers and specialize the resulting geometry to the independently tilting, radially arranged Holocopter and Voliro architectures.
\end{enumerate}

The remainder of this work is organized as follows. \Cref{sec:model} defines the platform geometry and nonlinear actuator map. \Cref{sec:metric} develops the actuator metric and wrench-fiber geometry. \Cref{sec:promptness} defines aerodynamic promptness and establishes its main properties. \Cref{sec:dfl} relates promptness to actuator dynamics and dynamic feedback linearization. \Cref{sec:shaping} develops task-preserving promptness shaping. \Cref{sec:radially_tilting_multirotors} specializes the theory to the Holocopter and Voliro, and \Cref{sec:conclusion} concludes this work.

\section{Multirotor Actuation Model}
\label{sec:model}

\subsection{Frames, geometry, and actuator state}

Consider a multirotor.
The inertial frame and the pose of the multirotor relative to it do
not enter the internal actuator-to-body-wrench map and are therefore
omitted from the theoretical development.

\begin{definition}[Platform frames]
Let \(\BodyFrame\) be a body-fixed frame whose origin is denoted by
\(\BodyOrigin\). Let \(\actuatorCount\) denote the number of
propellers and define
\(\actuatorIndexSet:=\{1,\ldots,\actuatorCount\}\). For each
\(\actuatorIndex\in\actuatorIndexSet\), let
\(\PropellerFrame{\actuatorIndex}\) be a frame attached to the stator
of the \(\actuatorIndex\)-th propeller unit. All wrench coordinates
used below are expressed in \(\BodyFrame\).
\end{definition}

\begin{definition}[Propeller geometry and orientation mechanism]
Let \(\OrientationConfigurationSpace\) be a smooth manifold of
dimension \(\configurationDimension\in\IntegerNumbers\), with
\(\configurationDimension\geq0\), describing the configuration of the
propeller-orientation mechanism. For each
\(\actuatorIndex\in\actuatorIndexSet\), let
\(\propellerPosition{\actuatorIndex}\in\RealNumbers^{\spatialDimension}\)
denote the position of \(\PropellerOrigin{\actuatorIndex}\) relative
to \(\BodyOrigin\), expressed in \(\BodyFrame\). Moreover, define the
smooth map
\[
\propellerDirectionMap{\actuatorIndex}:
\OrientationConfigurationSpace
\to
\UnitSphere,
\quad
\configurationState
\mapsto
\propellerDirectionMap{\actuatorIndex}(\configurationState),
\]
where
\(\propellerDirectionMap{\actuatorIndex}(\configurationState)\)
is the unit thrust direction of the
\(\actuatorIndex\)-th propeller at the mechanism configuration
\(\configurationState\), expressed in \(\BodyFrame\).

The \(\actuatorIndex\)-th propeller is fixed if
\(\propellerDirectionMap{\actuatorIndex}\) is constant. In general,
\(\propellerDirectionMap{\actuatorIndex}(\configurationState)\) may depend on any subset of the mechanism variables. Shared dependence of several propeller-direction maps on the same mechanism variables represents kinematic coupling among the corresponding propeller orientations.
\end{definition}

\begin{remark}[Independent one-axis tilting]
The independent one-axis tilting architecture is recovered with \(\OrientationConfigurationSpace=\TiltTorus\), \(\configurationDimension=\actuatorCount\), \(\configurationState=\tiltAngleArray\), and
\(
\propellerDirectionMap{\actuatorIndex}(\configurationState)
=
\rotationAboutAxis{\tiltAxis{\actuatorIndex}}{\tiltAngle{\actuatorIndex}}
\nominalThrustDirection{\actuatorIndex}
\).
More generally, a smooth transmission map \(\tiltTransmissionMap:\OrientationConfigurationSpace\to\TiltTorus\) may specify all physical propeller angles from fewer mechanism coordinates. Constant components describe fixed propellers, while repeated, signed, or nonlinear components describe shared motors, gearing, or linkages.
\end{remark}

\begin{assumption}[Ideal first-order propeller law]
\label{ass:propeller}
Each propeller produces thrust and aerodynamic reaction moment according to
\[
\thrustMagnitude{\actuatorIndex}=\thrustCoefficient{\actuatorIndex}\signedSquare{\propellerSpeed{\actuatorIndex}},\quad
\dragMomentMagnitude{\actuatorIndex}=\spinDirection{\actuatorIndex}\dragCoefficient{\actuatorIndex}\signedSquare{\propellerSpeed{\actuatorIndex}},
\]
where \(\propellerSpeed{\actuatorIndex}\in\RealNumbers\) is the actual signed propeller angular speed, \(\thrustCoefficient{\actuatorIndex}\in\PositiveReals\), \(\dragCoefficient{\actuatorIndex}\in\NonnegativeReals\), and \(\spinDirection{\actuatorIndex}\in\{-1,1\}\) specifies the aerodynamic reaction-moment orientation relative to the thrust.
\end{assumption}

\begin{remark}[Actual speeds are essential]
The choice of \(\propellerSpeed{\actuatorIndex}\), rather than \(\signedSquare{\propellerSpeed{\actuatorIndex}}\), as an actuator-state coordinate preserves the signed-quadratic nonlinearity and makes the available wrench-rate authority explicitly dependent on propeller speed. This dependence contributes to what aerodynamic promptness must measure.
\end{remark}

\begin{definition}[Actuator manifold]
Let \(\PropellerSpeedSpace\subseteq\RealNumbers^{\actuatorCount}\) be an open propeller-speed domain. The actuator manifold is
\[
\ActuatorStateSpace:=\PropellerSpeedSpace\times\OrientationConfigurationSpace
\ni
\actuatorState:=\bigl(\propellerSpeedArray,\configurationState\bigr),
\]
where \(\propellerSpeedArray\in\PropellerSpeedSpace\) and \(\configurationState\in\OrientationConfigurationSpace\). Hardware bounds define an admissible subset \(\AdmissibleActuatorStateSpace\subseteq\ActuatorStateSpace\).
\end{definition}

\subsection{Nonlinear actuator-to-wrench map}

\begin{definition}[Unit signed-squared-speed wrench map]
For each propeller, define the smooth map \(\unitWrenchCurve{\actuatorIndex}:\OrientationConfigurationSpace\to\RealNumbers^{\fullWrenchDimension}\) by
\begin{equation}
\unitWrenchCurve{\actuatorIndex}(\configurationState):=
\begin{bmatrix}
\thrustCoefficient{\actuatorIndex}\propellerDirectionMap{\actuatorIndex}(\configurationState)\\
\thrustCoefficient{\actuatorIndex}\crossProductMatrix{\propellerPosition{\actuatorIndex}}\propellerDirectionMap{\actuatorIndex}(\configurationState)
+\spinDirection{\actuatorIndex}\dragCoefficient{\actuatorIndex}\propellerDirectionMap{\actuatorIndex}(\configurationState)
\end{bmatrix}.
\label{eq:unit_wrench_curve}
\end{equation}
\end{definition}

\begin{definition}[Physical actuator map]
Let \(\WrenchSpace\subseteq\RealNumbers^{\wrenchDimension}\), with \(\wrenchDimension\leq\fullWrenchDimension\), be an open task-wrench space. The physical actuator map is the continuously differentiable map \(\actuatorMap:\ActuatorStateSpace\to\WrenchSpace\) defined, after selecting the relevant wrench coordinates, by
\begin{equation}
\bodyWrench=\actuatorMap(\propellerSpeedArray,\configurationState):=
\sum_{\actuatorIndex\in\actuatorIndexSet}
\unitWrenchCurve{\actuatorIndex}(\configurationState)\signedSquare{\propellerSpeed{\actuatorIndex}}.
\label{eq:physical_map}
\end{equation}
\end{definition}

\begin{remark}[Generality]
Equation~\eqref{eq:physical_map} covers arbitrary mixtures of fixed and tilting propellers, independent or mechanically coupled tilt mechanisms, heterogeneous propeller locations and coefficients, reversible and nonreversible operating subsets, and any propeller count. The platform-specific orientation mechanism is entirely described by \(\OrientationConfigurationSpace\) and the maps \(\propellerDirectionMap{\actuatorIndex}\).
\end{remark}

\begin{lemma}[Differential of the actuator map]
\label{lem:differential}
The physical actuator map is continuously differentiable. At every
\(\actuatorState\in\ActuatorStateSpace\), its differential is the
linear map
\[
\actuatorMapDifferentialAt{\actuatorState}:
\ActuatorTangentSpaceAt{\actuatorState}
\to
\WrenchTangentSpaceAt{\actuatorMap(\actuatorState)}.
\]
Its matrix in product actuator coordinates and the selected wrench
coordinates is
\begin{equation}
\actuatorMapJacobian(\actuatorState)
=
\begin{bmatrix}
\speedJacobian(\actuatorState)
&
\configurationJacobian(\actuatorState)
\end{bmatrix},
\label{eq:jacobian_partition}
\end{equation}
where the \(\actuatorIndex\)-th propeller-speed column is
\begin{equation}
[\speedJacobian(\actuatorState)]_{\actuatorIndex}
=
2\absoluteValue{\propellerSpeed{\actuatorIndex}}
\unitWrenchCurve{\actuatorIndex}(\configurationState),
\label{eq:speed_jacobian_columns}
\end{equation}
and the configuration block \(\configurationJacobian(\actuatorState)\) represents the linear map
\begin{equation}
\configurationRate
\longmapsto
\sum_{\actuatorIndex\in\actuatorIndexSet}
\signedSquare{\propellerSpeed{\actuatorIndex}}
{\differential\unitWrenchCurve{\actuatorIndex}}_{\configurationState}
(\configurationRate).
\label{eq:configuration_jacobian}
\end{equation}
Consequently, every actuator-state velocity satisfies
\begin{equation}
\taskRate
=
\speedJacobian(\actuatorState)\propellerAccelerationArray
+
\configurationJacobian(\actuatorState)\configurationRate.
\label{eq:wrench_rate}
\end{equation}
\end{lemma}

\begin{proof}
The signed-quadratic map
\(\SignedSquareMap(s)=s\absoluteValue{s}\) is continuously
differentiable, with
\(\SignedSquareMapDerivative(s)=2\absoluteValue{s}\).
Differentiating~\eqref{eq:physical_map} with respect to the propeller
speeds and the mechanism configuration yields
\eqref{eq:speed_jacobian_columns}
and~\eqref{eq:configuration_jacobian}, respectively.
Equation~\eqref{eq:wrench_rate} is the corresponding
product-coordinate representation of the differential.
\end{proof}

\begin{remark}[Fixed, independent, and coupled orientations]
If \(\propellerDirectionMap{\actuatorIndex}\) is constant, the
\(\actuatorIndex\)-th propeller contributes no term to the
configuration differential. If several propeller-direction maps
depend on the same mechanism coordinates, their contributions are
combined by
\(\configurationJacobian(\actuatorState)\) along the corresponding
mechanism-velocity directions. Hence, the same configuration
Jacobian represents fixed propellers, independently tilting
propellers, and mechanically coupled orientation mechanisms without
introducing individual propeller-angle coordinates.
\end{remark}

\begin{remark}[Different speed scalings]
The contribution of the \(\actuatorIndex\)-th propeller to
\(\speedJacobian(\actuatorState)\) scales linearly with
\(\absoluteValue{\propellerSpeed{\actuatorIndex}}\), whereas its
contribution to \(\configurationJacobian(\actuatorState)\) is weighted
by \(\signedSquare{\propellerSpeed{\actuatorIndex}}\). Consequently,
a propeller at zero speed provides neither propeller-acceleration nor
orientation-mechanism wrench-rate authority to first order. As the
propeller-speed magnitude increases, the wrench-rate effect of changing
its orientation grows quadratically, while that of changing its speed
grows linearly. Hence, the relative contribution of the orientation
mechanism may increase with propeller speed, although its actual
importance also depends on the available propeller accelerations,
mechanism rates, and the actuator-rate metric introduced below.
\end{remark}

\subsection{Local actuation and redundancy}

The following terminology is relative to the selected wrench task
\(\actuatorMap:\ActuatorStateSpace\to\WrenchSpace\) and to an actuator
state \(\actuatorState\in\ActuatorStateSpace\).

\begin{definition}[Local full actuation and underactuation]
The actuator system is locally fully actuated at
\(\actuatorState\) if
\[
\actuatorMapDifferentialAt{\actuatorState}:
\ActuatorTangentSpaceAt{\actuatorState}
\to
\WrenchTangentSpaceAt{\actuatorMap(\actuatorState)}
\]
is surjective. It is locally underactuated at
\(\actuatorState\) otherwise.
\end{definition}

Local full actuation means that every wrench-rate vector at the
current wrench can be generated by an actuator-state velocity. By the
local submersion theorem, the image under \(\actuatorMap\) of every
sufficiently small neighborhood of a locally fully actuated state
contains a neighborhood of the current wrench. The locally attainable
wrench set is therefore full-dimensional in \(\WrenchSpace\). This
statement applies to interior actuator states; hardware bounds may
restrict the admissible actuator-state velocities at boundary points.

\begin{definition}[Local redundancy]
The actuator system is locally redundant relative to the selected
wrench task at \(\actuatorState\) if
\[
\kernel
\actuatorMapDifferentialAt{\actuatorState}
\neq
\{0\}.
\]
Equivalently, there exists a nonzero actuator-state velocity that
produces zero instantaneous wrench rate.
\end{definition}

Local full actuation and local redundancy are independent properties.
At a locally fully actuated state,
\[
\dim
\kernel
\actuatorMapDifferentialAt{\actuatorState}
=
\dim\ActuatorStateSpace-\dim\WrenchSpace.
\]
The kernel directions are tangent to the current wrench fiber and
constitute the task-preserving internal directions studied below.

\section{Metric and Fiber Geometry}
\label{sec:metric}

\subsection{Actuator-rate metric}

\begin{definition}[Actuator Riemannian metric]
A Riemannian metric on \(\ActuatorStateSpace\) is a family of symmetric positive-definite bilinear forms \(\metric_{\actuatorState}:\ActuatorTangentSpaceAt{\actuatorState}\times\ActuatorTangentSpaceAt{\actuatorState}\to\RealNumbers\) depending smoothly on \(\actuatorState\). In local product coordinates induced by \((\propellerSpeedArray,\configurationState)\), the metric is represented by the smooth matrix-valued map
\[
\metricMatrix(\actuatorState):=
\begin{bmatrix}
\speedMetricMatrix(\actuatorState)&\couplingMetricMatrix(\actuatorState)\\
\couplingMetricMatrix(\actuatorState)^{\transpose}&\configurationMetricMatrix(\actuatorState)
\end{bmatrix}>0,
\]
where the block dimensions are \(\actuatorCount\) and \(\configurationDimension\). Accordingly, for \(\actuatorRate=(\propellerAccelerationArray,\configurationRate)\),
\[
\vectorNorm{\actuatorRate}_{\metric}^{2}:=
\actuatorRate^{\transpose}\metricMatrix(\actuatorState)\actuatorRate.
\]
\end{definition}

\begin{definition}[Independent ellipsoidal actuator-rate metric]
\label{def:diagonal_metric}
Let
\(\speedCapabilityMatrix:
\ActuatorStateSpace
\to
\RealNumbers^{\actuatorCount\times\actuatorCount}\)
and
\(\configurationCapabilityMatrix:
\ActuatorStateSpace
\to
\RealNumbers^{\configurationDimension\times\configurationDimension}\)
be smooth diagonal matrix-valued maps with strictly positive diagonal
entries. Their entries represent characteristic propeller
accelerations and orientation-mechanism rates, respectively. The
independent ellipsoidal actuator-rate metric is defined by
\[
\vectorNorm{\actuatorRate}_{\metric}^{2}
=
\vectorNorm{
\speedCapabilityMatrix(\actuatorState)^{-1}
\propellerAccelerationArray
}^{2}
+
\vectorNorm{
\configurationCapabilityMatrix(\actuatorState)^{-1}
\configurationRate
}^{2}.
\]
Its matrix in the selected product coordinates is
\begin{equation}
\metricMatrix(\actuatorState)
=
\diagOperator\!\left(
\speedCapabilityMatrix(\actuatorState)^{-\transpose}
\speedCapabilityMatrix(\actuatorState)^{-1},
\configurationCapabilityMatrix(\actuatorState)^{-\transpose}
\configurationCapabilityMatrix(\actuatorState)^{-1}
\right).
\label{eq:capability_metric}
\end{equation}
Because the capability matrices are diagonal with positive entries,
they are symmetric and invertible, and \eqref{eq:capability_metric}
reduces to
\[
\metricMatrix(\actuatorState)
=
\diagOperator\!\left(
\speedCapabilityMatrix(\actuatorState)^{-2},
\configurationCapabilityMatrix(\actuatorState)^{-2}
\right),
\]
which is symmetric positive definite.
\end{definition}

\begin{remark}[Role and limitations of the capability metric]
Aerodynamic promptness can be defined using any smooth Riemannian metric on \(\ActuatorStateSpace\); it does not require the particular metric in \Cref{def:diagonal_metric}. The independent ellipsoidal metric provides a physically interpretable choice whose unit ball approximates the available propeller accelerations and orientation-mechanism rates. Its block-diagonal structure expresses the absence of metric coupling between these two types of actuator motion and yields the additive decomposition of task-rate authority derived below. General smooth positive-definite metrics, including metrics with such coupling, remain admissible.

The coordinate representation applies on any mechanism chart and is global whenever the selected mechanism coordinates are global, as in the toroidal specializations considered later. A shared orientation mechanism contributes one mechanism-rate coordinate, independently of the number of propellers whose directions it changes, and its capability is therefore counted once. State dependence may represent reduced propeller acceleration near voltage or speed limits, configuration-dependent transmission ratios, load-dependent mechanism rates, or conservative thermal derating. Hard, asymmetric, or nonsmooth rate constraints require a more general state-dependent admissible-rate set.
\end{remark}

\subsection{Wrench fibers and horizontal directions}

\begin{definition}[Wrench fiber]
For \(\bodyWrench\in\WrenchSpace\), the actuator fiber is
\[
\fiberAt{\bodyWrench}
:=
\actuatorMap^{-1}(\bodyWrench)
=
\setBuilder{
\actuatorState\in\ActuatorStateSpace
}{
\actuatorMap(\actuatorState)=\bodyWrench
}.
\]
The feasible fiber is
\(\fiberAt{\bodyWrench}\cap\AdmissibleActuatorStateSpace\).
\end{definition}

\begin{definition}[Regular actuator state]
The regular actuator set is
\[
\RegularActuatorSet
:=
\setBuilder{
\actuatorState\in\ActuatorStateSpace
}{
\actuatorMapDifferentialAt{\actuatorState}:
\ActuatorTangentSpaceAt{\actuatorState}
\to
\WrenchTangentSpaceAt{\actuatorMap(\actuatorState)}
\text{ is surjective}
}.
\]
Equivalently, \(\actuatorState\in\RegularActuatorSet\) if and only if
\[
\matrixRank
\actuatorMapJacobian(\actuatorState)
=
\wrenchDimension.
\]
\end{definition}

\begin{proposition}[Local structure of a regular fiber]
\label{prop:regular_fiber}
If \(\actuatorState\in\RegularActuatorSet\), then
\(\fiberAt{\actuatorMap(\actuatorState)}\) is locally an embedded submanifold of class \(C^{1}\) of \(\ActuatorStateSpace\) of dimension
\(\actuatorCount+\configurationDimension-\wrenchDimension\), and
\[
T_{\actuatorState}
\fiberAt{\actuatorMap(\actuatorState)}
=
\kernel
\actuatorMapDifferentialAt{\actuatorState}.
\]
\end{proposition}

\begin{proof}
The result follows from the regular-level-set theorem applied to
\(\actuatorMap:\ActuatorStateSpace\to\WrenchSpace\).
\end{proof}

\begin{definition}[Vertical and horizontal spaces]
At \(\actuatorState\in\RegularActuatorSet\), the vertical space is
\[
\verticalSpaceAt{\actuatorState}
:=
\kernel
\actuatorMapDifferentialAt{\actuatorState}
=
T_{\actuatorState}
\fiberAt{\actuatorMap(\actuatorState)}.
\]
The metric-horizontal space is its
\(\metric_{\actuatorState}\)-orthogonal complement,
\[
\horizontalSpaceAt{\actuatorState}
:=
\verticalSpaceAt{\actuatorState}^{\orthogonalComplement,\metric}.
\]
Consequently,
\[
\ActuatorTangentSpaceAt{\actuatorState}
=
\verticalSpaceAt{\actuatorState}
\directSum
\horizontalSpaceAt{\actuatorState}.
\]
\end{definition}

\begin{definition}[Task-rate authority map]
\label{def:authority_map}
Let \(\actuatorState\in\ActuatorStateSpace\). The dual of the
actuator-map differential is the linear map
\[
\actuatorMapDifferentialAt{\actuatorState}^{*}:
\WrenchCotangentSpaceAt{\actuatorMap(\actuatorState)}
\to
\ActuatorCotangentSpaceAt{\actuatorState}.
\]
Let
\[
\musicalSharp_{\metric,\actuatorState}:
\ActuatorCotangentSpaceAt{\actuatorState}
\to
\ActuatorTangentSpaceAt{\actuatorState}
\]
be the musical isomorphism induced by
\(\metric_{\actuatorState}\).

The task-rate authority map at \(\actuatorState\) is
\[
\authorityMap_{\actuatorState}
:=
\actuatorMapDifferentialAt{\actuatorState}
\circ
\musicalSharp_{\metric,\actuatorState}
\circ
\actuatorMapDifferentialAt{\actuatorState}^{*},
\]
with
\[
\authorityMap_{\actuatorState}:
\WrenchCotangentSpaceAt{\actuatorMap(\actuatorState)}
\to
\WrenchTangentSpaceAt{\actuatorMap(\actuatorState)}.
\]
\end{definition}

\begin{proposition}[Metric horizontal inverse]
\label{prop:horizontal_inverse}
For every \(\actuatorState\in\RegularActuatorSet\), the authority map
\(\authorityMap_{\actuatorState}\) is an isomorphism. Moreover, the
restriction
\[
\left.
\actuatorMapDifferentialAt{\actuatorState}
\right|_{\horizontalSpaceAt{\actuatorState}}
:
\horizontalSpaceAt{\actuatorState}
\to
\WrenchTangentSpaceAt{\actuatorMap(\actuatorState)}
\]
is an isomorphism.

Its inverse is the linear map
\[
\horizontalInverseMapAt{\actuatorState}
:
\WrenchTangentSpaceAt{\actuatorMap(\actuatorState)}
\to
\horizontalSpaceAt{\actuatorState}
\]
defined by
\begin{equation}
\horizontalInverseMapAt{\actuatorState}
:=
\musicalSharp_{\metric,\actuatorState}
\circ
\actuatorMapDifferentialAt{\actuatorState}^{*}
\circ
\authorityMap_{\actuatorState}^{-1}.
\label{eq:intrinsic_horizontal_inverse}
\end{equation}

For every
\(\taskRate\in
\WrenchTangentSpaceAt{\actuatorMap(\actuatorState)}\),
the actuator-state velocity
\(\horizontalInverseMapAt{\actuatorState}(\taskRate)\) is the unique
element of \(\horizontalSpaceAt{\actuatorState}\) mapped to
\(\taskRate\) by
\(\actuatorMapDifferentialAt{\actuatorState}\). It is also the unique
minimum-\(\metric_{\actuatorState}\)-norm element among all
actuator-state velocities mapped to \(\taskRate\).

In the product actuator coordinates and the selected wrench
coordinates, the matrices representing
\(\authorityMap_{\actuatorState}\) and
\(\horizontalInverseMapAt{\actuatorState}\) are
\begin{equation}
\taskRateAuthorityMatrix(\actuatorState)
=
\actuatorMapJacobian(\actuatorState)
\metricMatrix(\actuatorState)^{-1}
\actuatorMapJacobian(\actuatorState)^{\transpose},
\label{eq:authority_matrix}
\end{equation}
and
\[
\horizontalInverseMatrix(\actuatorState)
=
\metricMatrix(\actuatorState)^{-1}
\actuatorMapJacobian(\actuatorState)^{\transpose}
\taskRateAuthorityMatrix(\actuatorState)^{-1},
\]
respectively.
\end{proposition}

\begin{proof}
Because
\(\actuatorMapDifferentialAt{\actuatorState}\) is surjective, its dual
\(\actuatorMapDifferentialAt{\actuatorState}^{*}\) is injective. Thus,
for every nonzero
\(\eta\in
\WrenchCotangentSpaceAt{\actuatorMap(\actuatorState)}\),
\[
\eta\bigl(\authorityMap_{\actuatorState}(\eta)\bigr)
=
\metric_{\actuatorState}^{-1}\!\left(
\actuatorMapDifferentialAt{\actuatorState}^{*}\eta,
\actuatorMapDifferentialAt{\actuatorState}^{*}\eta
\right)
>0.
\]
Hence, \(\authorityMap_{\actuatorState}\) is injective. Its domain and
codomain have the same finite dimension, so it is an isomorphism.

The image of
\(\musicalSharp_{\metric,\actuatorState}
\circ
\actuatorMapDifferentialAt{\actuatorState}^{*}\)
equals \(\horizontalSpaceAt{\actuatorState}\). Indeed, for every
\(\eta\in
\WrenchCotangentSpaceAt{\actuatorMap(\actuatorState)}\) and every
\(z\in\verticalSpaceAt{\actuatorState}\),
\[
\metric_{\actuatorState}\!\left(
\musicalSharp_{\metric,\actuatorState}
\actuatorMapDifferentialAt{\actuatorState}^{*}\eta,
z
\right)
=
\eta\!\left(
\actuatorMapDifferentialAt{\actuatorState}z
\right)
=
0.
\]
The image is therefore contained in the horizontal space. The two
spaces have dimension \(\wrenchDimension\), so they coincide.

It follows directly from
\Cref{eq:intrinsic_horizontal_inverse} that
\[
\actuatorMapDifferentialAt{\actuatorState}
\circ
\horizontalInverseMapAt{\actuatorState}
=
\operatorname{id}_{
\WrenchTangentSpaceAt{\actuatorMap(\actuatorState)}
}.
\]
Therefore,
\(\horizontalInverseMapAt{\actuatorState}\) is the inverse of the
restriction of
\(\actuatorMapDifferentialAt{\actuatorState}\) to the horizontal
space.

Any
\(\actuatorRate\in
\ActuatorTangentSpaceAt{\actuatorState}\) satisfying
\(\actuatorMapDifferentialAt{\actuatorState}
\actuatorRate=\taskRate\)
can be written uniquely as
\[
\actuatorRate
=
\horizontalInverseMapAt{\actuatorState}(\taskRate)
+
\verticalActuatorRate,
\quad
\verticalActuatorRate\in
\verticalSpaceAt{\actuatorState}.
\]
Orthogonality gives
\[
\vectorNorm{\actuatorRate}_{\metric}^{2}
=
\vectorNorm{
\horizontalInverseMapAt{\actuatorState}(\taskRate)
}_{\metric}^{2}
+
\vectorNorm{\verticalActuatorRate}_{\metric}^{2}.
\]
The norm is therefore minimized uniquely when
\(\verticalActuatorRate=0\). The coordinate expressions follow from
the matrices representing the differential, its dual, and the musical
isomorphism.
\end{proof}

\begin{proposition}[Metric-orthogonal vertical projector]
\label{prop:vertical_projector}
For every \(\actuatorState\in\RegularActuatorSet\), define
\[
\verticalProjectorMapAt{\actuatorState}
:
\ActuatorTangentSpaceAt{\actuatorState}
\to
\ActuatorTangentSpaceAt{\actuatorState}
\]
by
\[
\verticalProjectorMapAt{\actuatorState}
:=
\operatorname{id}_{\ActuatorTangentSpaceAt{\actuatorState}}
-
\horizontalInverseMapAt{\actuatorState}
\circ
\actuatorMapDifferentialAt{\actuatorState}.
\]
Then
\(\verticalProjectorMapAt{\actuatorState}\) is the
\(\metric_{\actuatorState}\)-orthogonal projector onto
\(\verticalSpaceAt{\actuatorState}\). Its kernel is
\(\horizontalSpaceAt{\actuatorState}\).

In the product actuator coordinates, its matrix is
\[
\verticalProjectorMatrix(\actuatorState)
=
\identityMatrix{\actuatorCount+\configurationDimension}
-
\horizontalInverseMatrix(\actuatorState)
\actuatorMapJacobian(\actuatorState).
\]
\end{proposition}

\begin{proof}
If \(z\in\verticalSpaceAt{\actuatorState}\), then
\(\actuatorMapDifferentialAt{\actuatorState}z=0\), and therefore
\(\verticalProjectorMapAt{\actuatorState}(z)=z\).

If \(z\in\horizontalSpaceAt{\actuatorState}\), the inverse property
in \Cref{prop:horizontal_inverse} gives
\[
\horizontalInverseMapAt{\actuatorState}
\bigl(
\actuatorMapDifferentialAt{\actuatorState}z
\bigr)
=
z,
\]
and therefore
\(\verticalProjectorMapAt{\actuatorState}(z)=0\).
The vertical-horizontal direct sum then proves that
\(\verticalProjectorMapAt{\actuatorState}\) is the projector onto the
vertical space along the horizontal space. Because these spaces are
\(\metric_{\actuatorState}\)-orthogonal, the projector is
metric-orthogonal. The coordinate expression follows by composition
of the corresponding matrices.
\end{proof}

\section{Aerodynamic Promptness}
\label{sec:promptness}

\begin{definition}[Task-rate authority bilinear form]
For each \(\actuatorState\in\ActuatorStateSpace\), the task-rate authority bilinear form is the map
\[
\taskRateAuthorityForm_{\actuatorState}:
\WrenchCotangentSpaceAt{\actuatorMap(\actuatorState)}
\times
\WrenchCotangentSpaceAt{\actuatorMap(\actuatorState)}
\to
\RealNumbers
\]
defined by
\[
\taskRateAuthorityForm_{\actuatorState}(\eta,\zeta)
:=
\eta\!\left(\authorityMap_{\actuatorState}(\zeta)\right)
=
\metric_{\actuatorState}^{-1}\!\left(
\actuatorMapDifferentialAt{\actuatorState}^{*}\eta,
\actuatorMapDifferentialAt{\actuatorState}^{*}\zeta
\right).
\]
Its coordinate matrix is \(\taskRateAuthorityMatrix(\actuatorState)\), defined in~\eqref{eq:authority_matrix}. The form is positive semidefinite at every actuator state and positive definite exactly on \(\RegularActuatorSet\).
\end{definition}

\begin{definition}[Aerodynamic promptness]
Let \(\taskMetric_{\bodyWrench}:\WrenchTangentSpaceAt{\bodyWrench}\times\WrenchTangentSpaceAt{\bodyWrench}\to\RealNumbers\) be a Riemannian metric on \(\WrenchSpace\), locally represented by \(\taskMetricMatrix(\bodyWrench)\). The metric-aware aerodynamic promptness is the map \(\metricWeightedPromptness:\ActuatorStateSpace\to\NonnegativeReals\) defined by
\[
\metricWeightedPromptness(\actuatorState):=
\sqrt{\matrixDeterminant\bigl(\taskRateAuthorityMatrix(\actuatorState)\taskMetricMatrix(\actuatorMap(\actuatorState))\bigr)}.
\]
In fixed task coordinates with the Euclidean task metric, the corresponding map \(\aerodynamicPromptness:\ActuatorStateSpace\to\NonnegativeReals\) is
\begin{equation}
\aerodynamicPromptness(\actuatorState):=
\sqrt{\matrixDeterminant\taskRateAuthorityMatrix(\actuatorState)}.
\label{eq:promptness}
\end{equation}
\end{definition}

This definition extends the generic fiber-density construction of~\cite{Franchi2026Promptness} to the product actuator manifold of configuration-dependent tilting-propeller multirotors. It measures the metric-weighted volume of wrench rates generated by the unit actuator-rate ball.

\begin{remark}[Promptness and readiness]
Aerodynamic promptness~\cite{Franchi2026Promptness} quantifies the wrench-rate readiness of an
actuator state. Accordingly, this work uses \emph{readiness} for the
physical capability and \emph{aerodynamic promptness} for its
metric-dependent volume measure. Subsequent work adopted this
readiness terminology to formulate a forward-invariant lower bound on
motor-only wrench-rate authority through a control barrier
function~\cite{Silano2026ReadinessBarrier}, and later extended that
construction to articulated multirotors by using active rotor tilt as
an authority-recovery resource~\cite{Silano2026TiltResource}.
\end{remark}

\begin{theorem}[Promptness and local wrench-rate surjectivity]
\label{thm:promptness_rank}
For every actuator state and every positive-definite actuator metric,
\[
\aerodynamicPromptness(\actuatorState)>0
\quad\Longleftrightarrow\quad
\matrixRank\actuatorMapJacobian(\actuatorState)=\wrenchDimension.
\]
\end{theorem}

\begin{proof}
Because \(\metricMatrix(\actuatorState)^{-1}\) is positive definite, \(\taskRateAuthorityMatrix(\actuatorState)\) and \(\actuatorMapJacobian(\actuatorState)\) have equal rank. The symmetric matrix \(\taskRateAuthorityMatrix(\actuatorState)\) is positive definite exactly when its rank equals \(\wrenchDimension\), which is equivalent to a strictly positive determinant.
\end{proof}

\begin{corollary}[Binary and quantitative authority]
The zero set of \(\aerodynamicPromptness\) is exactly the critical set of the nonlinear actuator map. Away from that set, the magnitude of \(\aerodynamicPromptness\) continuously quantifies the metric-normalized wrench-rate volume and therefore refines the binary rank condition.
\end{corollary}

\begin{proposition}[Propeller-acceleration and orientation-mechanism contributions]
\label{prop:additive_authority}
Under \Cref{def:diagonal_metric},
\begin{equation}
\taskRateAuthorityMatrix(\actuatorState)=\speedAuthorityContribution(\actuatorState)+\configurationAuthorityContribution(\actuatorState),
\label{eq:authority_sum}
\end{equation}
where
\begin{equation}
\begin{aligned}
\speedAuthorityContribution(\actuatorState)&:=\speedJacobian(\actuatorState)\speedCapabilityMatrix(\actuatorState)^{2}\speedJacobian(\actuatorState)^{\transpose},\\
\configurationAuthorityContribution(\actuatorState)&:=\configurationJacobian(\actuatorState)\configurationCapabilityMatrix(\actuatorState)^{2}\configurationJacobian(\actuatorState)^{\transpose}.
\label{eq:authority_components}
\end{aligned}
\end{equation}
Both contributions are positive semidefinite.
\end{proposition}

\begin{proof}
Insert~\eqref{eq:capability_metric} and~\eqref{eq:jacobian_partition} into~\eqref{eq:authority_matrix}. The off-diagonal blocks vanish, yielding~\eqref{eq:authority_sum}. Each term in~\eqref{eq:authority_components} is a Gram matrix.
\end{proof}

\begin{remark}[What is additive]
Propeller-acceleration and orientation-mechanism contributions add at the bilinear-form level, not at the scalar promptness level. In general, \(\sqrt{\matrixDeterminant(\speedAuthorityContribution(\actuatorState)+\configurationAuthorityContribution(\actuatorState))}\) cannot be separated into a sum of two scalar promptness measures.
\end{remark}

\begin{proposition}[Invariance of fiberwise promptness maximizers]
\label{prop:invariance}
For a fixed wrench \(\bodyWrench\), the maximizers of \(\metricWeightedPromptness\) over \(\fiberAt{\bodyWrench}\) are invariant under smooth changes of actuator coordinates, smooth changes of wrench coordinates, and the choice of task metric \(\taskMetric\).
\end{proposition}

\begin{proof}
The bilinear form \(\taskRateAuthorityForm\) is defined independently of coordinates, so actuator-coordinate changes affect only its representation. A wrench-coordinate change multiplies the coordinate volume by the absolute determinant of the coordinate-change Jacobian evaluated at \(\bodyWrench\). Likewise, \(\taskMetric\) contributes the positive factor \(\sqrt{\matrixDeterminant\taskMetricMatrix(\bodyWrench)}\). Both are constant on \(\fiberAt{\bodyWrench}\) and therefore preserve its maximizers, as established for the generic construction in~\cite{Franchi2026Promptness}.
\end{proof}

\begin{definition}[Promptness-optimal section]
For each regular wrench, the set of promptness-optimal realizations is
\begin{equation}
\argmaximum_{\actuatorState\in\fiberAt{\bodyWrench}\cap\AdmissibleActuatorStateSpace}
\aerodynamicPromptness(\actuatorState).
\label{eq:promptness_optimization}
\end{equation}
Let \(\mathcal{O}\subseteq\WrenchSpace\) be open and contained in the set of regular values of \(\actuatorMap\). A promptness-optimal local section is a continuously differentiable map \(\chi:\mathcal{O}\to\AdmissibleActuatorStateSpace\) satisfying \(\actuatorMap\circ\chi=\operatorname{id}_{\mathcal{O}}\) and \(\chi(\bodyWrench)\) belonging to the set in~\eqref{eq:promptness_optimization} for every \(\bodyWrench\in\mathcal{O}\).
\end{definition}

\section{Connection with Dynamic Feedback Linearization}
\label{sec:dfl}

\subsection{General actuator dynamics}

Let \(\InputSpace\) be a finite-dimensional real input space. Let
\(\actuatorDrift:\ActuatorStateSpace\to T\ActuatorStateSpace\) be a
smooth vector field, so that
\(\actuatorDrift(\actuatorState)
\in\ActuatorTangentSpaceAt{\actuatorState}\) for every
\(\actuatorState\in\ActuatorStateSpace\).

Let
\[
\actuatorInputMap:
\ActuatorStateSpace\times\InputSpace
\to
T\ActuatorStateSpace
\]
be smooth, satisfy
\(\actuatorInputMap(\actuatorState,\actuatorVelocityInput)
\in\ActuatorTangentSpaceAt{\actuatorState}\), and be linear in
\(\actuatorVelocityInput\) for every fixed \(\actuatorState\).
Accordingly,
\[
\actuatorInputMapAt{\actuatorState}:
\InputSpace
\to
\ActuatorTangentSpaceAt{\actuatorState},
\quad
\actuatorInputMapAt{\actuatorState}(\actuatorVelocityInput)
:=
\actuatorInputMap(\actuatorState,\actuatorVelocityInput),
\]
is a linear map. The actuator dynamics are
\begin{equation}
\actuatorRate
=
\actuatorDrift(\actuatorState)
+
\actuatorInputMapAt{\actuatorState}(\actuatorVelocityInput).
\label{eq:actuator_dynamics}
\end{equation}

In selected actuator and input coordinates, let
\(\actuatorInputMatrix(\actuatorState)\) denote the matrix representing
\(\actuatorInputMapAt{\actuatorState}\). Equation~\eqref{eq:actuator_dynamics}
then becomes
\[
\actuatorRate
=
\actuatorDrift(\actuatorState)
+
\actuatorInputMatrix(\actuatorState)\actuatorVelocityInput.
\]

\begin{definition}[Dynamic-extension decoupling map]
For every \(\actuatorState\in\ActuatorStateSpace\), the decoupling map
is the linear map
\[
\decouplingMapAt{\actuatorState}
:
\InputSpace
\to
\WrenchTangentSpaceAt{\actuatorMap(\actuatorState)}
\]
defined by
\[
\decouplingMapAt{\actuatorState}
:=
\actuatorMapDifferentialAt{\actuatorState}
\circ
\actuatorInputMapAt{\actuatorState}.
\]
Its matrix in the selected input and wrench coordinates is
\[
\decouplingMatrix(\actuatorState)
:=
\actuatorMapJacobian(\actuatorState)
\actuatorInputMatrix(\actuatorState).
\]
\end{definition}

The corresponding wrench-rate dynamics are
\[
\taskRate
=
\actuatorMapDifferentialAt{\actuatorState}
\bigl(\actuatorDrift(\actuatorState)\bigr)
+
\decouplingMapAt{\actuatorState}(\actuatorVelocityInput),
\]
or, in the selected coordinates,
\[
\taskRate
=
\actuatorMapJacobian(\actuatorState)
\actuatorDrift(\actuatorState)
+
\decouplingMatrix(\actuatorState)
\actuatorVelocityInput.
\]

\begin{definition}[Input metric]
An input metric is a positive-definite bilinear form
\[
\inputMetric:
\InputSpace\times\InputSpace
\to
\RealNumbers.
\]
Let
\(\inputMusicalSharp:\InputSpace^{*}\to\InputSpace\) be the
corresponding musical isomorphism. In the selected input coordinates,
let \(\inputMetricMatrix\) be the positive-definite matrix representing
\(\inputMetric\); the matrix representing \(\inputMusicalSharp\) is
\(\inputMetricMatrix^{-1}\).
\end{definition}

\begin{definition}[Input-induced actuator-rate authority map]
For each \(\actuatorState\in\ActuatorStateSpace\), the actuator input
map and the input metric define the linear map
\[
\inputActuatorAuthorityMapAt{\actuatorState}
:
\ActuatorCotangentSpaceAt{\actuatorState}
\to
\ActuatorTangentSpaceAt{\actuatorState}
\]
given by
\[
\inputActuatorAuthorityMapAt{\actuatorState}
:=
\actuatorInputMapAt{\actuatorState}
\circ
\inputMusicalSharp
\circ
\actuatorInputMapAt{\actuatorState}^{*}.
\]
Its coordinate matrix is
\[
\inputActuatorAuthorityMatrix(\actuatorState)
:=
\actuatorInputMatrix(\actuatorState)
\inputMetricMatrix^{-1}
\actuatorInputMatrix(\actuatorState)^{\transpose}.
\]
\end{definition}

The map
\(\inputActuatorAuthorityMapAt{\actuatorState}\) is positive
semidefinite in the sense that
\[
\eta\!\left(
\inputActuatorAuthorityMapAt{\actuatorState}(\eta)
\right)
=
\inputMetric^{-1}\!\left(
\actuatorInputMapAt{\actuatorState}^{*}\eta,
\actuatorInputMapAt{\actuatorState}^{*}\eta
\right)
\geq 0
\]
for every
\(\eta\in\ActuatorCotangentSpaceAt{\actuatorState}\).

If
\(\actuatorInputMapAt{\actuatorState}:
\InputSpace\to\ActuatorTangentSpaceAt{\actuatorState}\)
is an isomorphism for every \(\actuatorState\), then
\(\inputActuatorAuthorityMapAt{\actuatorState}\) is an isomorphism.
Its inverse defines a Riemannian metric on
\(\ActuatorStateSpace\). Otherwise, its image is
\(\image\actuatorInputMapAt{\actuatorState}\), and it represents only
the actuator-state velocities generated through the input map.

\begin{theorem}[Input-induced task-rate authority]
\label{thm:dfl}
For every \(\actuatorState\in\ActuatorStateSpace\), the decoupling map
and the input metric define the input-induced task-rate authority map
\[
\dflAuthorityMapAt{\actuatorState}
:
\WrenchCotangentSpaceAt{\actuatorMap(\actuatorState)}
\to
\WrenchTangentSpaceAt{\actuatorMap(\actuatorState)}
\]
given by
\[
\dflAuthorityMapAt{\actuatorState}
:=
\decouplingMapAt{\actuatorState}
\circ
\inputMusicalSharp
\circ
\decouplingMapAt{\actuatorState}^{*}.
\]
Equivalently,
\[
\dflAuthorityMapAt{\actuatorState}
=
\actuatorMapDifferentialAt{\actuatorState}
\circ
\inputActuatorAuthorityMapAt{\actuatorState}
\circ
\actuatorMapDifferentialAt{\actuatorState}^{*}.
\]

The matrix representing
\(\dflAuthorityMapAt{\actuatorState}\) in the selected wrench
coordinates is
\[
\dflAuthorityMatrix(\actuatorState)
:=
\decouplingMatrix(\actuatorState)
\inputMetricMatrix^{-1}
\decouplingMatrix(\actuatorState)^{\transpose}.
\]
Consequently,
\begin{equation}
\dflPromptness(\actuatorState)
:=
\sqrt{
\matrixDeterminant
\dflAuthorityMatrix(\actuatorState)
}
>0
\quad\Longleftrightarrow\quad
\matrixRank
\decouplingMatrix(\actuatorState)
=
\wrenchDimension.
\label{eq:dfl_equivalence}
\end{equation}
\end{theorem}

\begin{proof}
For every
\(\eta\in
\WrenchCotangentSpaceAt{\actuatorMap(\actuatorState)}\),
\[
\eta\!\left(
\dflAuthorityMapAt{\actuatorState}(\eta)
\right)
=
\inputMetric^{-1}\!\left(
\decouplingMapAt{\actuatorState}^{*}\eta,
\decouplingMapAt{\actuatorState}^{*}\eta
\right).
\]
Because \(\inputMetric^{-1}\) is positive definite, the kernel of
\(\dflAuthorityMapAt{\actuatorState}\) equals the kernel of
\(\decouplingMapAt{\actuatorState}^{*}\). Therefore,
\(\dflAuthorityMapAt{\actuatorState}\) is an isomorphism if and only if
\(\decouplingMapAt{\actuatorState}\) is surjective. In coordinates,
this is equivalent to
\(\matrixRank\decouplingMatrix(\actuatorState)=\wrenchDimension\),
which in turn is equivalent to the positivity of the determinant in
\eqref{eq:dfl_equivalence}.
\end{proof}

\begin{corollary}[Direct actuator-rate inputs]
\label{cor:direct_actuator_rate_inputs}
At a fixed \(\actuatorState\in\ActuatorStateSpace\), suppose that \(\InputSpace\) is identified with \(\ActuatorTangentSpaceAt{\actuatorState}\), \(\actuatorInputMapAt{\actuatorState}\) is the identity map,
\(\actuatorDrift(\actuatorState)=0\), and
\(\inputMetric=\metric_{\actuatorState}\). Then
\[
\decouplingMapAt{\actuatorState}
=
\actuatorMapDifferentialAt{\actuatorState},
\quad
\dflAuthorityMapAt{\actuatorState}
=
\authorityMap_{\actuatorState}.
\]
In the corresponding coordinates,
\[
\decouplingMatrix(\actuatorState)
=
\actuatorMapJacobian(\actuatorState),
\quad
\dflAuthorityMatrix(\actuatorState)
=
\taskRateAuthorityMatrix(\actuatorState).
\]
Dynamic feedback linearization is therefore feasible exactly where
aerodynamic promptness is nonzero.
\end{corollary}

\begin{remark}[Required qualification]
For general actuator dynamics, positive promptness of the unrestricted
actuator map \(\actuatorMap\) is necessary but not sufficient for
dynamic feedback linearization. The image of
\(\actuatorInputMapAt{\actuatorState}\) must contain enough
actuator-state directions for
\(\decouplingMapAt{\actuatorState}
=
\actuatorMapDifferentialAt{\actuatorState}
\circ
\actuatorInputMapAt{\actuatorState}\)
to be surjective. Exact equivalence holds for the input-induced
promptness \(\dflPromptness\) in~\eqref{eq:dfl_equivalence}.
\end{remark}

\subsection{Conditioning interpretation}

\begin{proposition}[Minimum normalized input effort]
Suppose that
\(\decouplingMapAt{\actuatorState}\) is surjective. Define
\[
\dflHorizontalInverseAt{\actuatorState}
:
\WrenchTangentSpaceAt{\actuatorMap(\actuatorState)}
\to
\InputSpace
\]
by
\[
\dflHorizontalInverseAt{\actuatorState}
:=
\inputMusicalSharp
\circ
\decouplingMapAt{\actuatorState}^{*}
\circ
\dflAuthorityMapAt{\actuatorState}^{-1}.
\]
For every
\(\driftCompensatedTaskRate
\in
\WrenchTangentSpaceAt{\actuatorMap(\actuatorState)}\),
the input
\[
\actuatorVelocityInput
=
\dflHorizontalInverseAt{\actuatorState}
(\driftCompensatedTaskRate)
\]
is the unique minimum-\(\inputMetric\)-norm input satisfying
\[
\decouplingMapAt{\actuatorState}
(\actuatorVelocityInput)
=
\driftCompensatedTaskRate.
\]
Its squared input norm is
\[
\vectorNorm{\actuatorVelocityInput}_{\inputMetric}^{2}
=
\driftCompensatedTaskRate^{\transpose}
\dflAuthorityMatrix(\actuatorState)^{-1}
\driftCompensatedTaskRate.
\]
\end{proposition}

\begin{proof}
This is the construction of
\Cref{prop:horizontal_inverse} applied to the linear map
\(\decouplingMapAt{\actuatorState}:
\InputSpace\to
\WrenchTangentSpaceAt{\actuatorMap(\actuatorState)}\)
with the input metric \(\inputMetric\).
\end{proof}

The rank condition on
\(\decouplingMatrix(\actuatorState)\) determines the boundary of the
dynamic-feedback-linearization domain. The spectrum of
\(\dflAuthorityMatrix(\actuatorState)\) quantifies the directional
input authority and its conditioning, whereas
\(\dflPromptness(\actuatorState)\) summarizes the corresponding
task-rate volume.

\section{Task-Preserving Promptness Shaping}
\label{sec:shaping}

\begin{definition}[Promptness-shaping actuator-rate law]
At \(\actuatorState\in\RegularActuatorSet\), for \(\desiredTaskRate\in\WrenchTangentSpaceAt{\actuatorMap(\actuatorState)}\) and \(\internalMotion\in\ActuatorTangentSpaceAt{\actuatorState}\), define
\begin{equation}
\actuatorRate
=
\horizontalInverseMapAt{\actuatorState}(\desiredTaskRate)
+
\verticalProjectorMapAt{\actuatorState}(\internalMotion).
\label{eq:general_shaping_law}
\end{equation}
The first term is horizontal and realizes the desired wrench rate; the second is vertical and changes the actuator state without changing the instantaneous wrench rate.
\end{definition}

\begin{proposition}[Exact task preservation]
Every actuator-rate law of the form~\eqref{eq:general_shaping_law} satisfies
\[
\actuatorMapDifferentialAt{\actuatorState}(\actuatorRate)=\desiredTaskRate.
\]
In the selected actuator and wrench coordinates, this identity is represented by \(\actuatorMapJacobian(\actuatorState)\actuatorRate=\desiredTaskRate\).
\end{proposition}

\begin{proof}
By \Cref{prop:horizontal_inverse,prop:vertical_projector}, \(\actuatorMapDifferentialAt{\actuatorState}\circ\horizontalInverseMapAt{\actuatorState}\) is the identity on \(\WrenchTangentSpaceAt{\actuatorMap(\actuatorState)}\), whereas \(\actuatorMapDifferentialAt{\actuatorState}\circ\verticalProjectorMapAt{\actuatorState}=0\).
\end{proof}

\begin{definition}[Promptness regulation]
Define \(\promptnessSquared:\ActuatorStateSpace\to\NonnegativeReals\) by \(\promptnessSquared(\actuatorState):=\aerodynamicPromptness(\actuatorState)^{2}=\matrixDeterminant\taskRateAuthorityMatrix(\actuatorState)\). Let \(\ActuatorStateSpace^{\times}\subseteq\ActuatorStateSpace\) denote the open subset on which every propeller speed is nonzero. On this subset, \(\actuatorMap\), \(\promptnessSquared\), and the Riemannian gradient \(\riemannianGradient\promptnessSquared:\ActuatorStateSpace^{\times}\to T\ActuatorStateSpace\) are smooth. The gradient is uniquely defined by \(\metric_{\actuatorState}((\riemannianGradient\promptnessSquared)(\actuatorState),z)=\differential(\promptnessSquared)_{\actuatorState}(z)\) for every \(z\in\ActuatorTangentSpaceAt{\actuatorState}\). A smooth task-preserving regulation law on \(\ActuatorStateSpace^{\times}\cap\RegularActuatorSet\) is obtained by choosing
\begin{equation}
\internalMotion:=\promptnessGain\bigl(\promptnessReference^{2}-\promptnessSquared(\actuatorState)\bigr)(\riemannianGradient\promptnessSquared)(\actuatorState),
\quad \promptnessGain\in\PositiveReals.
\label{eq:promptness_regulator}
\end{equation}
Using squared promptness avoids differentiating the square root in~\eqref{eq:promptness}; the feedback law is considered on \(\ActuatorStateSpace^{\times}\cap\RegularActuatorSet\), where the required derivatives and the vertical projector are well defined.
\end{definition}

\begin{proposition}[Fiberwise promptness regulation]
For a constant wrench command, the purely vertical law obtained
from~\eqref{eq:general_shaping_law}
and~\eqref{eq:promptness_regulator} satisfies
\begin{equation}
\frac{\differential}{\differential t}
\promptnessSquared(\actuatorState)
=
\promptnessGain
\bigl(
\promptnessReference^{2}
-
\promptnessSquared(\actuatorState)
\bigr)
\vectorNorm{
\verticalProjectorMapAt{\actuatorState}
\bigl(
(\riemannianGradient\promptnessSquared)(\actuatorState)
\bigr)
}_{\metric}^{2}.
\label{eq:promptness_evolution}
\end{equation}
Consequently, squared promptness moves monotonically toward the squared reference
whenever its fiber-restricted gradient is nonzero. Every closed-loop
equilibrium satisfies either
\[
\promptnessSquared(\actuatorState)
=
\promptnessReference^{2}
\]
or
\[
\verticalProjectorMapAt{\actuatorState}
\bigl(
(\riemannianGradient\promptnessSquared)(\actuatorState)
\bigr)
=
0.
\]
\end{proposition}

\begin{proof}
Evaluate the differential of \(\promptnessSquared\) along the vertical
component of~\eqref{eq:general_shaping_law}. The defining property of
the Riemannian gradient and the metric self-adjointness of
\(\verticalProjectorMapAt{\actuatorState}\) yield
\eqref{eq:promptness_evolution}. The equilibrium characterization
follows because \(\promptnessGain>0\).
\end{proof}

\begin{remark}[Regulation rather than maximization]
A finite reference allows the controller to retain only the readiness required by the task, rather than always maximizing promptness and potentially incurring unnecessary propulsive effort, tilt-actuator activity, or proximity to hardware limits.
\end{remark}

\begin{theorem}[Preservation of actuator-map regularity]
Let \(\mathcal{K}\subset\AdmissibleActuatorStateSpace\cap\ActuatorStateSpace^{\times}\) be compact and satisfy \(\aerodynamicPromptness(\actuatorState)\geq\rho_{0}\) for every \(\actuatorState\in\mathcal{K}\), where \(\rho_{0}\in\PositiveReals\). If the commanded wrench rate and the vertical promptness regulator render \(\mathcal{K}\) forward invariant, then the actuator map remains regular and the metric horizontal inverse in~\eqref{eq:general_shaping_law} remains well defined along every trajectory in \(\mathcal{K}\). Under the direct actuator-rate-input conditions of \Cref{cor:direct_actuator_rate_inputs}, dynamic feedback linearization also remains feasible.
\end{theorem}

\begin{proof}
The positive lower bound and \Cref{thm:promptness_rank} imply full row rank of \(\actuatorMapJacobian(\actuatorState)\) throughout \(\mathcal{K}\). Hence, \(\taskRateAuthorityMatrix(\actuatorState)\) is invertible and the horizontal inverse exists at every state of the invariant set.
\end{proof}

\section{Specialization to Radially Tilting Multirotors}
\label{sec:radially_tilting_multirotors}

\subsection{Radially tilting architecture}

We now specialize the general formulation to multirotors whose
\(\actuatorCount\geq4\) propeller units are uniformly distributed on a
circle and tilt independently about the corresponding radial arms.
Let \(\armLength\) denote the circle radius and define the azimuth of
the \(\actuatorIndex\)-th propeller by
\[
\radialAzimuth{\actuatorIndex}
:=
2\pi(\actuatorIndex-1)
\;/\;
\actuatorCount.
\]
Its position and tilt axis are
\begin{equation}
\propellerPosition{\actuatorIndex}
=
\rotationAboutAxis{\canonicalAxis{3}}
{\radialAzimuth{\actuatorIndex}}
\begin{bmatrix}
\armLength & 0 & 0
\end{bmatrix}^{\transpose},
\quad
\tiltAxis{\actuatorIndex}
=
\propellerPosition{\actuatorIndex}
\;/\;
\armLength.
\label{eq:radial_tilting_geometry}
\end{equation}

Because all propeller tilts are independent and periodic, the
orientation-mechanism configuration is described by
\[
\OrientationConfigurationSpace
=
\TiltTorus,
\quad
\configurationDimension
=
\actuatorCount,
\quad
\configurationState
=
\tiltAngleArray.
\]
Let
\(\nominalThrustDirection{\actuatorIndex}\in\UnitSphere\) denote the
unit thrust direction of the \(\actuatorIndex\)-th propeller at zero
tilt, expressed in \(\BodyFrame\). Its direction map is
\[
\propellerDirectionMap{\actuatorIndex}:
\TiltTorus
\to
\UnitSphere,
\quad
\tiltAngleArray
\mapsto
\rotationAboutAxis{\tiltAxis{\actuatorIndex}}
{\tiltAngle{\actuatorIndex}}
\nominalThrustDirection{\actuatorIndex}.
\]

The nominal thrust directions and aerodynamic reaction-moment signs
depend on the frame and propeller conventions of the specific
platform. Substitution of these quantities and
\eqref{eq:radial_tilting_geometry} into
\eqref{eq:unit_wrench_curve} defines the corresponding radially
tilting actuator map through~\eqref{eq:physical_map}. The
Holocopter~\cite{Ryll2015Holocopter} and
Voliro~\cite{Kamel2018Voliro} are instances of this architecture.

\begin{corollary}[Regular full-wrench fibers]
\label{cor:radial_tilting_fiber_dimension}
At every regular actuator state of a radially tilting multirotor, the
complete-wrench fiber has local dimension
\[
2\actuatorCount-6.
\]
\end{corollary}

\begin{proof}
For the radially tilting architecture,
\(\configurationDimension=\actuatorCount\). The general regular-fiber
dimension
\(\actuatorCount+\configurationDimension-\fullWrenchDimension\)
therefore becomes \(2\actuatorCount-\fullWrenchDimension\).
\end{proof}

\subsection{The radially tilting quadrotor}

The Holocopter of~\cite{Ryll2015Holocopter} has
\(\actuatorCount=\holocopterRotorCount\) propeller units placed on four
orthogonal radial arms. Its geometry follows
\eqref{eq:radial_tilting_geometry}, with
\[
\radialAzimuth{\actuatorIndex}
=
(\actuatorIndex-1)\pi/2.
\]
The nominal thrust directions and propeller-group frame conventions
are those of~\cite{Ryll2015Holocopter}. Their substitution into
\eqref{eq:unit_wrench_curve} recovers the trigonometric force and
moment matrices of the original model.

By \Cref{cor:radial_tilting_fiber_dimension}, every regular
full-wrench fiber of the Holocopter is locally an embedded
submanifold of class \(C^{1}\) and dimension two.

\subsubsection{Relation to the Holocopter dynamic extension}

The original Holocopter dynamic extension uses the signed squared
propeller speeds and the tilt angles as internal actuator coordinates.
Let
\[
\varphi:
(\RealNumbers\setminus\{0\})^{\actuatorCount}
\to
(\RealNumbers\setminus\{0\})^{\actuatorCount}
\]
be defined componentwise by
\[
[\varphi(\propellerSpeedArray)]_{\actuatorIndex}
=
\signedSquare{\propellerSpeed{\actuatorIndex}},
\]
and define the product-coordinate transformation
\[
\Phi:
(\RealNumbers\setminus\{0\})^{\actuatorCount}
\times
\TiltTorus
\to
(\RealNumbers\setminus\{0\})^{\actuatorCount}
\times
\TiltTorus,
\quad
(\propellerSpeedArray,\tiltAngleArray)
\mapsto
\bigl(
\varphi(\propellerSpeedArray),
\tiltAngleArray
\bigr).
\]
On each fixed speed-sign orthant, \(\Phi\) is a diffeomorphism and its
differential is represented by
\[
J_{\Phi}(\actuatorState)
=
\diagOperator\!\left(
2\absoluteValue{\propellerSpeed{1}},
\ldots,
2\absoluteValue{\propellerSpeed{\actuatorCount}},
\identityMatrix{\actuatorCount}
\right).
\]

Let
\[
\widehat{\actuatorMap}:
(\RealNumbers\setminus\{0\})^{\actuatorCount}
\times
\TiltTorus
\to
\WrenchSpace
\]
denote the Holocopter actuator map expressed in signed-squared-speed
and tilt-angle coordinates. The physical actuator map in
\eqref{eq:physical_map} then satisfies
\begin{equation}
\actuatorMap
=
\widehat{\actuatorMap}
\circ
\Phi.
\label{eq:holocopter_map_factorization}
\end{equation}

\begin{proposition}[Coordinate relation to the Holocopter decoupling map]
\label{prop:holocopter_coordinate_relation}
Away from zero propeller speeds, the Jacobians of the physical and
signed-squared-speed actuator maps satisfy
\begin{equation}
\actuatorMapJacobian(\actuatorState)
=
J_{\widehat{\actuatorMap}}
\bigl(
\Phi(\actuatorState)
\bigr)
J_{\Phi}(\actuatorState).
\label{eq:holocopter_jacobian_relation}
\end{equation}
Consequently,
\[
\matrixRank
\actuatorMapJacobian(\actuatorState)
=
\matrixRank
J_{\widehat{\actuatorMap}}
\bigl(
\Phi(\actuatorState)
\bigr),
\]
and
\[
\kernel
\actuatorMapJacobian(\actuatorState)
=
J_{\Phi}(\actuatorState)^{-1}
\kernel
J_{\widehat{\actuatorMap}}
\bigl(
\Phi(\actuatorState)
\bigr).
\]
Thus, the physical actuator-map differential and the actuator block
of the original Holocopter decoupling map have the same rank, while
their task-preserving directions are related by the tangent
coordinate transformation \(J_{\Phi}(\actuatorState)\).
\end{proposition}

\begin{proof}
Differentiating the factorization
in~\eqref{eq:holocopter_map_factorization} gives
\eqref{eq:holocopter_jacobian_relation}. On each fixed speed-sign
orthant, \(J_{\Phi}(\actuatorState)\) is invertible because every
propeller speed is nonzero. Right multiplication by this invertible
matrix preserves rank and maps the kernel according to the stated
relation.
\end{proof}

The original dynamic extension differentiates the rigid-body output
until the actuator-coordinate rates enter explicitly. Its decoupling
matrix therefore contains
\(J_{\widehat{\actuatorMap}}\) composed with the nonsingular
rigid-body and output-coordinate maps introduced by the selected
output derivative~\cite{Ryll2015Holocopter}. These additional factors
do not change surjectivity. Hence, away from zero propeller speeds,
the established full-row-rank condition of the Holocopter decoupling
matrix is equivalent to surjectivity of
\(\actuatorMapDifferentialAt{\actuatorState}\).

\begin{corollary}[Geometric interpretation of the Holocopter rank condition]
\label{cor:holocopter_promptness}
Under the ideal Holocopter model and away from zero propeller speeds,
the established surjectivity of the original decoupling
map~\cite{Ryll2015Holocopter} implies that the physical actuator map
is locally fully actuated. By
\Cref{thm:promptness_rank}, its aerodynamic promptness is positive for
every Riemannian actuator-rate metric.

Moreover, because
\(\configurationDimension=\actuatorCount=\holocopterRotorCount\),
the corresponding regular full-wrench fiber has dimension two by
\Cref{cor:radial_tilting_fiber_dimension}. The two-dimensional
null space used by the original controller is the coordinate
representation of the tangent space to this wrench fiber.
\end{corollary}

When the direct actuator-rate-input model applies, the input map is
the identity and the input metric coincides with the actuator-rate
metric. The input-induced authority map of
\Cref{sec:dfl} then reduces to the aerodynamic-promptness authority
map. The original Holocopter rank condition determines whether every
wrench-rate direction is locally available, whereas aerodynamic
promptness and its spectrum quantify the volume, directional
distribution, and conditioning of that authority. Task-preserving
motion along the two-dimensional wrench fiber can therefore be used
to shape these quantities without changing the current wrench.

\subsection{The radially tilting hexarotor}

The original Voliro platform of~\cite{Kamel2018Voliro} has
\(\actuatorCount=\voliroRotorCount\) propeller units uniformly
distributed on a circle. Each propeller tilts independently about the
corresponding radial arm. By
\Cref{cor:radial_tilting_fiber_dimension}, every regular full-wrench
fiber of Voliro is locally an embedded submanifold of class \(C^{1}\)
and dimension six.

Thus, at a regular full-wrench realization, Voliro retains six
task-preserving internal actuator directions, whereas the Holocopter
retains two. These additional internal directions are available for
shaping promptness, actuator margins, and proximity to mechanical
limits, and for incorporating a separate propulsive-power objective
when an appropriate power model is supplied.

\subsubsection{Relation to the original Voliro allocation}

On the original nonnegative-speed operating set, the Voliro
allocation introduces, for each propeller, the vertical and lateral
force components
\[
\voliroVerticalForce{\actuatorIndex}
=
\thrustCoefficient{\actuatorIndex}
\signedSquare{\propellerSpeed{\actuatorIndex}}
\cos\tiltAngle{\actuatorIndex},
\quad
\voliroLateralForce{\actuatorIndex}
=
\thrustCoefficient{\actuatorIndex}
\signedSquare{\propellerSpeed{\actuatorIndex}}
\sin\tiltAngle{\actuatorIndex}.
\]
Stacking these components defines
\(\voliroDecomposedForceArray\in
\RealNumbers^{2\actuatorCount}\). The body wrench is represented by
the static linear map
\[
\voliroStaticAllocationMap:
\RealNumbers^{2\actuatorCount}
\to
\WrenchSpace,
\]
whose matrix in the selected decomposed-force and wrench coordinates
is \(\voliroStaticAllocationMatrix\). Hence,
\(\bodyWrench
=\voliroStaticAllocationMap(\voliroDecomposedForceArray)\), represented
in coordinates by
\(\bodyWrench
=\voliroStaticAllocationMatrix\voliroDecomposedForceArray\).

The original allocator selects the Moore--Penrose solution of this
linear equation and reconstructs each propeller speed and tilt angle
from the corresponding pair of vertical and lateral force
components~\cite{Kamel2018Voliro}.

\begin{proposition}[Relation between decomposed-force and physical actuator maps]
\label{prop:voliro_coordinate_relation}
Let
\[
\voliroDecompositionMap:
\PropellerSpeedSpace\times\TiltTorus
\to
\RealNumbers^{2\actuatorCount}
\]
be the continuously differentiable map that associates each physical
Voliro actuator state with the corresponding stacked array of
vertical and lateral propeller-force components. Then the Voliro
body-wrench map factors as
\[
\actuatorMap
=
\voliroStaticAllocationMap
\circ
\voliroDecompositionMap,
\]
and its differential satisfies
\[
\actuatorMapDifferentialAt{\actuatorState}
=
\voliroStaticAllocationMap
\circ
\differential\voliroDecompositionMap_{\actuatorState}.
\]
Consequently, the static allocation map and the physical
actuator-map differential describe different relations: the former
maps decomposed propeller-force components to body wrench, whereas the
latter maps physical actuator-state velocities to body-wrench rate.
\end{proposition}

\begin{proof}
The factorization follows from the definitions of the vertical and
lateral force components and from their linear contributions to the
body wrench. Differentiating the factorization gives the stated
relation by the chain rule.
\end{proof}

\begin{remark}[Relation to the Voliro minimum-norm allocation]
The squared Euclidean norm of the decomposed-force array used in the
original Voliro allocation is proportional to
\[
\sum_{\actuatorIndex\in\actuatorIndexSet}
\propellerSpeed{\actuatorIndex}^{4}.
\]
The Moore--Penrose allocation therefore minimizes a
decomposed-force norm, which is a propeller-loading surrogate rather
than an exact electrical-power objective. This objective concerns the
current wrench realization, whereas aerodynamic promptness quantifies
the metric-normalized ability to change that wrench.

The six-dimensional Voliro wrench fiber permits these objectives to
be studied jointly. In particular, task-preserving actuator motion
can modify the decomposed-force objective and full-wrench reaction
capability while retaining the same instantaneous body wrench.
\end{remark}

\begin{remark}[Tilt-rate limitations]
The original Voliro allocation recovers propeller speeds and tilt
angles from independently allocated vertical and lateral force
components. It does not directly account for the slower dynamics of
the orientation mechanisms relative to the propeller motors, a
difference identified as a practical limitation of the original
controller~\cite{Kamel2018Voliro}.

The actuator-rate metric incorporates this difference into the
differential wrench authority. Assigning different characteristic
scales to propeller acceleration and tilt rate prevents rapidly
changing propeller forces and comparatively slow tilt motion from
being treated as equally available actuator-state velocities.
\end{remark}

\section{Conclusion}
\label{sec:conclusion}

This work formulated tilting-propeller multirotor actuation as a nonlinear map from actual signed propeller speeds and a general propeller-orientation mechanism configuration to body wrench. A smooth Riemannian metric on the actuator manifold compares propeller accelerations and orientation-mechanism rates, and its musical isomorphism converts the dual actuator-map differential into physically normalized actuator-state velocities. Composition with the actuator-map differential yields a task-rate authority map from wrench covectors to wrench-rate vectors. Its coordinate matrix is positive semidefinite at every actuator state and positive definite exactly at regular points; the square root of its determinant defines aerodynamic promptness in fixed Euclidean wrench coordinates.

The same geometry separates task realization from task-preserving internal motion. At regular states, the actuator-map differential restricts to an isomorphism from the metric-horizontal space to the wrench tangent space. Its inverse produces the unique minimum-metric-norm actuator-state velocity for a prescribed wrench rate, while the complementary orthogonal projector selects motions tangent to the current wrench fiber. Projected gradient motion can therefore regulate promptness without altering the instantaneous wrench.

For general actuator dynamics, the intrinsic decoupling map is the actuator-map differential composed with the actuator input map. The associated input metric induces a task-rate authority map whose positivity is equivalent to surjectivity of the decoupling map, recovering the dynamic-feedback-linearization rank condition while additionally quantifying directional authority and conditioning. The independently tilting, radially arranged Holocopter and Voliro illustrate how the same framework encompasses architectures with two- and six-dimensional regular full-wrench fibers, respectively, and clarifies the relation between their historical allocation constructions and the differential geometry of the physical actuator map.

The resulting framework identifies wrench-rate readiness as a task-preserving property of redundant tilting actuation. It provides the geometric objects required to compare actuator states that generate the same wrench, quantify their complete differential authority, and regulate that authority without disturbing the current task.

\bibliographystyle{IEEEtran}
\bibliography{refs}

\end{document}